\documentclass[a4paper]{article}

\usepackage[title]{appendix}
\usepackage{hyperref}
\usepackage{graphicx}
\usepackage{amssymb}
\usepackage{amsthm}
\usepackage{amsmath}
\usepackage{epstopdf}
\usepackage{times}
\usepackage{mathrsfs}
\usepackage{tikz}
\usetikzlibrary{shapes,arrows,positioning,shapes.geometric,backgrounds}
\usepackage{ifthen}
\usepackage{algorithm}
\usepackage[noend]{algpseudocode}
\usepackage{relsize}

\newcommand{\tT}{\tilde{T}}

\newcommand{\cF}{\mathcal{F}}

\newcommand{\X}{\mathbb{X}}
\newcommand{\cX}{\mathcal{X}}
\newcommand{\Z}{\mathbb{Z}}
\newcommand{\N}{\mathbb{N}}

\newcommand{\defeq}{:=}
\newcommand{\ind}{\mathbb{I}}

\newcommand{\SPR}{\sigma}

\newcommand{\Y}{\mathbb{Y}}
\newcommand{\ch}{\mathrm{ch}}
\newcommand{\pa}{\mathrm{pa}}

\newcommand{\canon}{\mathbf{C}}

\newcommand{\ts}{\mathbf{t}}

\newcommand{\cN}{\mathcal{N}}

\renewcommand{\SPR}{\mathrm{SPR}}

\newcommand{\bprune}{u}
\newcommand{\bregrf}{v}
\newcommand{\tprune}{r}
\newcommand{\tregrf}{w}

\newcommand{\indistsh}[1]{\xrightarrow[\phantom{iii}]{{\mathrm{d}}}}

\newcommand{\almostsurelyOneArgSh}[1]{\xrightarrow[\phantom{iii}]{\mathrm{a.s.}}}

\algnewcommand\Break{\textbf{break}}
\algnewcommand\Continue{\textbf{continue}}

\definecolor{lightgray}{rgb}{0.85, 0.85, 0.85}
\definecolor{myblack}{RGB}{2,2,2}
\definecolor{myred}{RGB}{244,42,54}
\definecolor{mygreen}{RGB}{1,140,81}
\definecolor{myblue}{RGB}{19,92,163}
\definecolor{myyellow}{RGB}{250,124,57}

\newcommand{\maxsite}{S}

\newcommand{\ud}{\mathrm{d}}

\newcommand\ls{{\mathcal L}}

\newtheorem{theorem}{Theorem}
\newtheorem{lemma}{Lemma}

\newcommand{\sprop}[4]{(\mathrm{#1},\mathrm{#2},\mathrm{#3},\mathrm{#4})}

\newcommand{\cY}{\mathcal{Y}}
\newcommand{\cZ}{\mathcal{Z}}

\newcommand{\ns}{m_{\mathrm{S}}}

\newcommand{\jmax}{J}

\newcommand{\cP}{\mathcal{P}}

\begin{document}

\title{Bridging trees for posterior inference on Ancestral Recombination Graphs}


\author{
K. Heine\footnote{Corresponding author, email:k.m.p.heine@bath.ac.uk, Department of Mathematical Sciences, University of Bath, Claverton Down, Bath, BA2 7AY, UK}, A. Beskos\footnote{Department of Statistical Science, University College London, Gower Street, London, WC1E 6BT, UK}, A. Jasra\footnote{Department of Statistics and Applied Probability, National University of Singapore, 6 Science Drive 2, Singapore, 117546, SG}, D. Balding\footnote{Centre for Systems Genomics, School of BioSciences and School of Mathematics \& Statistics, University of Melbourne, Vic 3010, Australia} ~and M. De Iorio\footnote{Department of Statistical Science, University College London, Gower Street, London, WC1E 6BT, UK, and Yale-NUS College, 16 College Avenue West, Singapore, 138527, SG}}


\date{}



\maketitle

\begin{abstract}
We present a new Markov chain Monte Carlo algorithm, implemented in software Arbores, for inferring the history of a sample of DNA sequences. Our principal innovation is a bridging procedure, previously applied only for simple stochastic processes, in which the local computations within a bridge can proceed independently of the rest of the DNA sequence, facilitating large-scale parallelisation.
\end{abstract}

\section{Introduction}

A central problem in population genetics is to infer 
genealogical histories over a set of homologous DNA sequences, including e.g.\@ shared lineages of genome segments back to their most recent common ancestors (MRCAs), or  recombination, mutation and other genomic events that underlie the observed sequence variation. A Bayesian approach
 judiciously  combines structured probability models for the evolution 
of genealogies backwards in time (i.e.\@ prior probabilities) with information in the DNA sequences to provide posterior probabilities of lineages. It also gives rise to the highly complex computational challenge of probing the deduced posterior distributions.
We make considerable progress towards achieving this goal
by developing a new Markov chain Monte Carlo (MCMC) algorithm, and accompanying software Arbores. Our methodology adopts 
a data augmentation direction and
 implements a bridging procedure to impute the latent sequence of genealogical trees between given trees at two genome sites.

A prominent probabilistic evolutionary model for DNA sequences is the \emph{coalescent with recombination} --- an extension of the single-locus (Kingman's) coalescent~\cite{kingman82b,kingman82a}. The \emph{ancestral recombination graph} (ARG)~\cite{griffiths_et_marjoram97,grif:96} is a graphical representation of this evolutionary model. Essentially, ARG is graph representation of gene genealogies and the centre of our attention in this work.
ARG is central in population genetics and, more generally, in biology, and describes the relationship between sequences undergoing recombination. Recombination is one of the most important evolutionary forces as it increases genetic diversity and promotes adaptation through exchange of genetic material~\cite{arenas13}. Conscientious modelling and 
learning of recombination is fundamental to gaining a better understanding of many biological processes, e.g.~genome structure \cite{arenas13}, phenotypic diversity \cite{zhang_et_al02}, or mapping of disease genes. 
See \cite{arenas13} for a review of possible applications of ARG-based models.
%
%
In Bayesian inference, ARG is necessary for determining the likelihood function of a sample of observed chromosomes in a population genetic model, and, as such, for parameter estimation and hypothesis testing.  Unfortunately, due to the complexity of the ARG representation and the computational difficulties associated to inferring an ARG given a sample of chromosomes, ARGs have not been widely used and inference is often based on summary statistics and univariate techniques, with inevitable loss of information.

More formally, a  coalescent with recombination model can be thought of as specifying a prior law for an ARG. 
We follow the original definition and treat the ARG as a random graph, whereas some authors use this term to refer to a fixed graph, usually the one representing the true but unknown underlying genealogy or an estimate of it.
The inference from an observed set of DNA sequences can then, in theory, proceed by deriving the corresponding posterior ARG distribution. 
We adopt a discrete approximation to the infinite-sites mutation model, which implies that one of at most two alleles can be found at any given genomic site (in practice, triallelic DNA sites exist but are rare). It follows that an observed set of $N$ homologous DNA sequences, each comprised of $\maxsite$ genomic sites,  can be represented as a binary matrix $D$ with $D_{i,j} \in \{0,1\}$ for $1\leq i\leq N$, $1\leq j\leq\maxsite$.  
%
 %
%
%
%
%
%
Inferring the ARG conditional on the data $D$ is notoriously difficult  and has been a renowned challenge in computational genetics. Firstly, the posterior distribution is highly concentrated compared to the prior, so that elementary approaches based on importance or rejection sampling from the prior distribution are unacceptably inefficient. Secondly, the ARG is defined on a large and complex space involving a high-dimensional product of continuous and finite spaces, such that an exhaustive iteration even over the finite spaces is infeasible. Standard MCMC methods can nowadays be usefully applied in some high dimensional models (e.g.\@ large hierarchical models, or inverse problems \cite{arenas13}), but small changes in the random variables that specify an ARG can give substantially different graphs (i.e.\@ the likelihood surface is highly discontinuous), so that standard random-walk type MCMC methods are very difficult to implement for ARG inference.


Earlier attempts to address this sampling  task have involved mainly importance sampling \cite{grif:96, fearnhead_et_donnelly01} or MCMC \cite{nielsen00,kuhner00}. In \cite{nielsen00}, the primary motivation is the estimation of the `global' recombination rate, rather than the full ARG, enabling the use of less informative observations which makes the problem substantially simpler than the one considered here. The method of \cite{kuhner00} is based on proposing a replacement for a part of the ARG in proportion to its prior probability. These 
methods scale badly with $S$ and $N$. More recently, \cite{mcvean_et_cardin05} introduced the sequentially-Markov coalescent (SMC) approximation of the ARG which assumes that the sequence of genealogical trees at genome sites forms a Markov process. The authors demonstrated that the approximation is accurate over a wide range of recombination rates.
The SMC, coupled with the observed sequences, can be viewed as a hidden Markov model, with coalescent trees as hidden states, from which the data $D$ are generated under the mutation model. The most recent advance in this field, \cite{rasmussen_et_al14},
exploits the Markovian structure of the SMC approximation and proceeds by re-simulating the sequence of coalescent subtrees that correspond to a specific subset of observed haplotypes. With some resemblance to the method of \cite{kuhner00}, under the Markovian assumption, the resulting algorithm is a Gibbs sampler capable of performing inference on a moderately sized sample of sequences representing the complete genome.

\subsection{Overview of the proposed method}

A key contribution of our methodology is to reduce the computational complexity of the MCMC algorithm by dividing the inference task into a large number of subtasks, amenable to parallel computations. For each pre-specified genome segment, the coalescent trees at the initial and terminal sites are fixed, while a new sequence of trees is proposed at the intervening sites, compatible with the data $D$ and the terminal trees. The proposed sequence is accepted or rejected according to the Metropolis-Hastings (MH) acceptance test. Due to the Markovian structure of the SMC process, these calculations require no information from outside the chosen genome segment, see Figure \ref{fig:plot_bridge}. This is crucial to the algorithm's ability to localise the computations and control the combinatorial complexity. 

\begin{figure}
\begin{center}
\includegraphics[width=.8\textwidth]{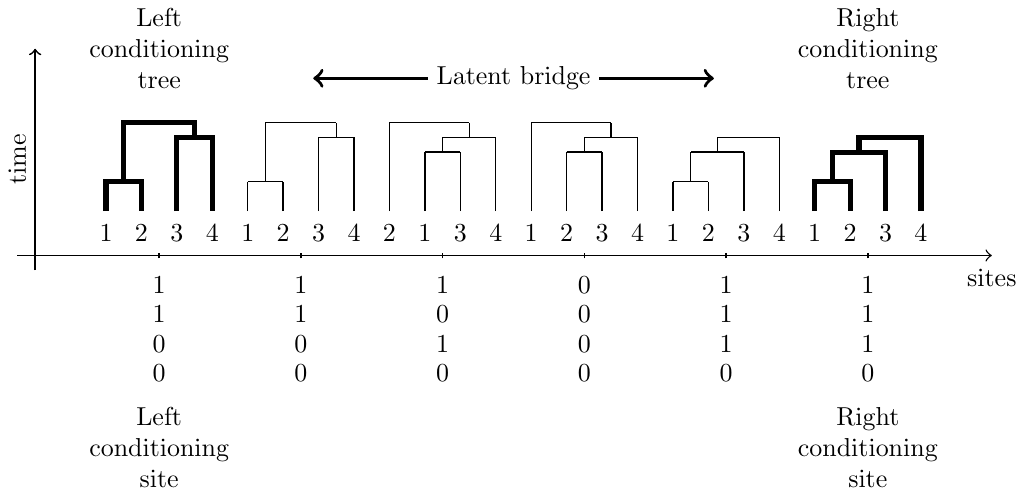}
\end{center}
\caption{The data augmentation scheme for the SMC model. Given left and right conditioning trees our method samples an intervening bridge of coalescent trees, 
which is consistent with the relevant observed sites. In this example, the trees are comprised of 7 nodes (4 of which, at the bottom of the graph, are the leaves).}
\label{fig:plot_bridge}
\end{figure}

This approach is similar in principle to bridging methods for discretely-observed diffusion processes (see e.g.~\cite{robert:01}) or Markov jump-processes (see e.g.~\cite{boys:08}). In recognition of this conceptual similarity, we refer to the proposed sequence of coalescent trees over a genome segment as a \emph{bridge}, and to our algorithm as the \emph{tree-bridging MCMC sampler}.  We are not aware of previous uses of bridging in such a complex state space. The major challenge for our algorithm involves developing an efficient way of generating bridges, which requires three main steps:
\begin{description}
\item[1. Tree Scanning:] We construct a bridge by a method inspired by the \emph{tree scanning} algorithm \cite{song_et_hein03,song_et_hein05}, \cite[page 328]{gusfield14}. Starting with the left conditioning tree (see Figure \ref{fig:plot_bridge}), the tree at each subsequent site is either the same, corresponding to an identity operation, or a subtree is cut off ({\em pruned}) and then reattached ({\em regrafted}) onto another branch, which is called a subtree-prune-and-regraft (SPR) operation \cite{song_et_hein03,song03,song06}. We iterate exhaustively over all such operations at each site, until all possible paths of coalescent trees have been generated up to the right conditioning site.

Tree scanning delivers paths of trees, each endowed with a total time ordering of the nodes (i.e.~the coalescences or branch merges), but the exact times associated with the nodes are not yet fixed. Moreover, most tree paths will typically not match the right terminal condition. 

\item[2. Time Adjustment:] Node times at the leftmost tree are fixed by the conditioning. Each non-identity SPR operation within a tree path introduces a new node. For each path, we check whether the times of the new nodes can be set so that the times in the right terminal tree are realised. If not, the path is discarded, otherwise the times of the nodes that are present in the right conditioning trees are fixed to their value in that conditioning tree.

\item[3. Sampling:] The last step is the generation of the complete bridge with all nodes affixed (i.e.\@ their times determined). We choose a coalescent tree path at random from the set resulting after Step 2. For this path, we generate times for each node not present in either terminal tree, from the uniform or exponential distribution (see later sections for details) in the permitted interval that respects the tree structures and node orderings.
\end{description}

\noindent The process is repeated for a number of bridges spanning the genome region, and overlapping so that each genome site is non-terminal in at least one bridge, thus is able to vary.

The remainder of this paper is organised as follows. Section \ref{sec:HMM} reviews the Markovian dynamics of the SMC model, and the corresponding likelihood of the observed DNA sequences. Section \ref{sec:Bridge sampler} presents the MCMC algorithm for sampling from the  ARG posterior distribution. Section \ref{sec:enabling heuristics} shows some heuristics for reducing the computational cost of the MCMC proposal. Section \ref{sec:numeric} gives numerical results from an example application of the new algorithm. We conclude with some remarks in Section \ref{sec:conclude}.
 
\section{Hidden Markov model for ARG inference}
\label{sec:HMM}

The adoption of the SMC approximation for the law of the ARG enables us to formulate the inference problem in the hidden Markov model framework as described in this section. \\ 


\subsection{Hidden tree process}

A realisation of the SMC process is a sequence (or path) of coalescent trees $T=(T_i)_{1\le i\le\maxsite}$, where each $T_i$ is a rooted binary tree. A tree branch is a directed edge identified by an ordered pair of nodes. Branches are named after the child node, so $(u,v)$ is referred to as branch $v$. Each tree has a unique \emph{root} node that has no parents. Any other node $v$ has a parent $\pa(v)$. The children of $v$ are denoted by $\ch(v)$. Associated with each observed sequence is a leaf node $v$ such that $\ch(v)=\emptyset$. We identify the nodes of each $T_i$ with integers $\{1,\ldots,2N{-}1\}$, the first $N$ nodes being the leaves. 
Each tree $T_i$ is fully specified by its topology $\canon_i$ and node times $\ts_i = (t_{i,1},\ldots,t_{i,2N{-}1})$ as exemplified in Figure \ref{fig:SPR}a. The topology $\canon_i$ can be represented as a $(N{-}1)\times 2$ matrix whose $k$th row includes the two elements of $\ch(N{+}k)$. The two elements of each row of $\canon_i$ are placed in increasing order. We have $t_{i,n}=0$ for the leaf nodes, while non-leaf nodes are indexed in increasing time order.

\begin{figure}
\begin{center}
\includegraphics[width=.9\textwidth]{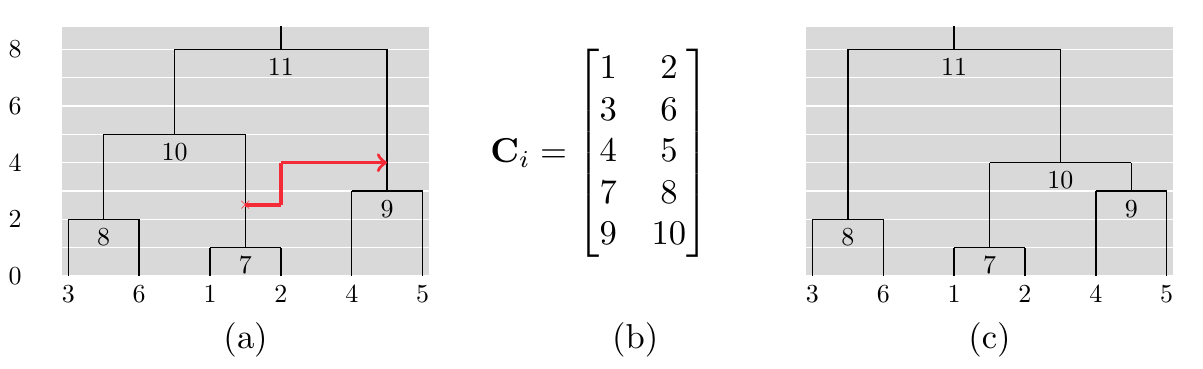}
\end{center}
\caption{(a) A graph representation of  a coalescent tree $T_{i}$ having the topology $\canon_{i}$ shown  in (b), with node times $\ts_{i} = (0,0,0,0,0,0,1,2,3,5,8)$, $u_{i} = 7$, $v_{i} = 9$, $r_{i} = 2.5$, $w_{i} = 4$. (c) The resulting tree after applying the SPR operation $(u_{i},v_{i},r_{i},w_{i})$ to the coalescent tree $T_{i}$.
}
\label{fig:SPR}
\end{figure}

Given $T_i=(\canon_i, \ts_i)$, for some $1 \leq i \leq S-1$, the next tree $T_{i+1}$ is determined by the occurrence (or not) of a recombination between sites $i$, $i+1$, and the
specification of such recombination. A recombination is represented here via an SPR operation characterised by a quadruple consisting of two nodes, $u_{i}$, $v_{i}$, and two positive real times $r_{i}$, $w_{i}$. If no recombination happens between sites $i$ and $i{+}1$ then by convention $(u_{i},v_{i}, r_{i}, w_{i})=(0,0,0,0)$, representing the identity SPR operation, and consequently $T_i=T_{i+1}$. If a recombination does occur, then the subtree\footnote{Without exception we use \emph{subtree} to refer to a tree whose leaves are also leaves in the original tree.} rooted at node $u_{i}$ is pruned from the tree at time $r_i$ and subsequently regrafted back into the tree onto branch $v_{i}$ at time $w_{i}>r_{i}$, as shown in Figure~\ref{fig:SPR}.

We  proceed to a detailed description of the dynamics of the hidden tree process. We need to provide a few  definitions. For measurable spaces $(\X,\cX)$, $(\Y,\cY)$, $(\Z,\cZ)$ and probability kernels $K_{1}:\X \times \cY \to [0,1]$ and $K_{2}:\Y \times \cZ \to [0,1]$ we define the probability kernel $K_{1}K_{2}:\X\times (\cY\otimes\cZ)$ so that for $x \in \X$, $A \in \cY\otimes\cZ$
\begin{equation*}
(K_{1}K_{2})(x,A) = \\
\int K_{1}(x,\ud y)\int K_{2}(y,\ud z) \ind_{A}(y,z).
\end{equation*}
For a kernel $K:\X\times \cY \to [0,1]$ and a probability measure $\mu:\cX\to [0,1]$ we define the probability measure $\mu K:\cX\otimes\cY\to[0,1]$ such that, for $A \in \cX \otimes \cY$
\begin{equation*}
(\mu K)(A) = \int \mu(\ud x)\int K(x,\ud y) \ind_{A}(x,y).
\end{equation*}
For probability measures $\mu_{1}$, $\mu_{2}$ defined on $(\X,\cX)$, $(\Y,\cY)$, respectively, we let $\mu_{1}\otimes\mu_{2}$ denote the product measure on $(\X\times\Y,\cX\otimes\cY)$. The $k$-fold iterate of a kernel $K$ is written as $K^{k}$.

The SMC process is initiated with a standard coalescent tree, followed by Markovian transitions defined via SPR operations. Formally, we define the initial distribution $\mu$ and the kernel $K$ via the coordinate-wise decompositions
\begin{align*}
\mu = \mu_{\canon,\ts}K_{\bprune}K_{\tprune}K_{\tregrf}K_{\bregrf}, \quad 
K = K_{\canon,\ts}K_{\bprune}K_{\tprune}K_{\tregrf}K_{\bregrf},
\end{align*}
from which we obtain the law of the SMC process as $\mu K^{S-1}$. We now define the probability measure $\mu_{\canon,\ts}$ and   kernels $K_{\canon,\ts}$, $K_{\bprune}$, $K_{\tprune}$, $K_{\tregrf}$, $K_{\bregrf}$. 
First, we note that $\mu_{\canon,\ts}$ simply corresponds to the law of the standard coalescent tree \cite{kingman82a,kingman82b,hein_et_al05} with a formal definition as follows. The law $\mu_{\canon}$ of the initial canonical form $\canon_{1}$ is defined such that for all $1 \leq i \leq N - 1$ the pair $(\canon_{1}(i,1),\canon_{1}(i,2))$ -- where $\canon_{1}(i,j)$ denotes the element on the $i$th row and $j$th column of $\canon_{1}$ -- is uniformly distributed over the set
\begin{equation*}
\big\{\,(j,k) \in \cN_{i-1}\times \cN_{i-1}: j < k\,\big\},
\end{equation*}
starting from the set of leaves $\cN_{0} = \{1,\ldots,N\}$, and then having, for $1 \leq i < N -1 $,
\begin{equation*}
\cN_{i} = \left(\cN_{i-1} \setminus \{\canon_{1}(i,1),\canon_{1}(i,2)\} \right)\cup \{N+i\}.
\end{equation*}
To define the law $\mu_{\ts}$ of the  coalescence times $\ts_{1}$, we write 
$f(x;\beta) = \beta e^{-\beta x}$, $x\geq 0$,
for the probability density of the exponential distribution with parameter $\beta>0$ and define $\mu_{\ts} = \mu_{t_{1}}\otimes \cdots \otimes \mu_{t_{2N-1}}$, where 
\begin{equation*}
\mu_{t_{i}}(\ud x) = 
\begin{cases}
\delta_{0}(\ud x), & 1\leq i \leq N\\
f\left(x;\beta_{0}\right)\ud x, & i = N+1\\ 
f\left(x-t_{i-1};\beta_{i-(N+1)}\right)\ud x, & i > N+1\\ 
\end{cases}
\end{equation*}
and $\beta_{i} = {N-i \choose 2}$, $0\leq i < N-1$. 
The initial distribution is simply $\mu_{\canon,\ts} = \mu_{\canon}\otimes \mu_{\ts}$.
Conditionally on a given coalescent tree $T_{i} = (\canon_{i},\ts_{i})$,  $1 \leq i \leq S-1$, with some probability there will be no recombination between sites $i$ and $i+1$, or there will be a recombination at a position chosen uniformly along the branches of the tree, with involved recombination rate parameter $\rho>0$. 
Formally, the conditional law of $\bprune_{i}$ is 
\begin{align*}
K_{\bprune}(T_{i},\ud x) 
&=  
\left(1-e^{-\rho L_{i}}\right)\sum_{n=1}^{2N-2}\frac{l_{i,n}}{L_{i}}\delta_{n}(\ud x) + e^{-\rho L_{i}}\delta_{0}(\ud x),
\end{align*}
where the total branch length $L_{i}$ and branch lengths $l_{i,n}$ of $T_{i}$ are defined as 
\begin{equation}\label{eq:tbl_and_bl}
L_i \defeq \sum_{n=1}^{2N-2} l_{i,n}.\quad\mbox{where }l_{i,n} \defeq t_{i,\pa(n)}-t_{i,n}.
\end{equation}
For a given coalescent tree $T_{i}$ and a pruned node $\bprune_{i}$, the conditional law of the pruning time $\tprune_i$ is defined as
\begin{align*}
K_{\tprune}(T_{i},\bprune_{i},\ud x) &= 
\ind(\bprune_{i}\neq 0)\,
\dfrac{\ind(x \in [t_{i,\bprune_{i}},~t_{i,\pa(\bprune_{i})}])}{l_{i,\bprune_{i}}}\ud x  +
\ind(\bprune_{i}=0) \delta_{0}(\ud x).
\end{align*}
The indicator functions $\ind(\bprune_{i}=0)$ and $\ind(\bprune_{i}\neq 0)$ specify the conditional distributions of $\tprune_{i}$ in the two cases: recombination occurs ($\bprune_{i}\neq 0$) or it does not ($\bprune_{i} = 0$). 
Next, the SMC model assumes that branch $\bprune_i$ is pruned at time $\tprune_i$, so the segment of the branch above this time is deleted. The remaining part of the separated branch is extended backwards in time from $\tprune_i$ and is re-attached to the main body of the tree according to the standard coalescent tree dynamics, i.e.\@ according to a Poisson process with rate equal to the number of existing branches
at any time instance. Formally, given $(T_{i},\bprune_{i},\tprune_{i})$, let $k$ be the number of nodes above the 
pruning time $\tprune_{i}$, including the parent node $\pa(\bprune_{i})$, and $\tilde{t}_1<\cdots < \tilde{t}_{k}$ their times 
with the convention $\tilde{t}_0 = \tprune_{i}$, $\tilde{t}_{k+1}=\infty$.
The conditional law of the coalescence time $\tregrf_{i}$ given $(T_{i},\bprune_{i},\tprune_{i})$ is 
\begin{align*}
K_{\tregrf}(T_{i},\bprune_{i},\tprune_{i},\ud x) &= \ind(\bprune_{i}\neq 0)f_{T_{i},\bprune_{i},\tprune_{i}}(x)\ud x  +
\ind(\bprune_{i}=0)\delta_{0}(\ud x),
\end{align*}
where, for $F(\,\cdot\,;\,\beta)$ denoting the cumulative distribution function of $\mathrm{Exp}(\beta)$, $\beta>0$, and  $\bar{F} = 1 - F$, 
\begin{align*}
f_{T_{i},\bprune_{i},\tprune_{i}}(x) &=  \sum_{j=0}^{k} f(x-\tilde{t}_j;\beta_{j})\,\mathbb{I}\,[\,x\in (\tilde{t}_j,\tilde{t}_{j+1})\,]  
\left[\prod_{l=0}^{j-1}\bar{F}(\tilde{t}_{l+1}-\tilde{t}_{l}; \beta_{l})\right],
\end{align*}
with
\begin{equation*}
\beta_{j} = \begin{cases}k+1-j,&\tilde{t}_{j}\geq t_{i,\pa(\bprune_{i})}\\ k-j, &\text{otherwise}.\end{cases}
\end{equation*}
The conditional law of the regraft node $\bregrf_{i}$ given $(T_i,\bprune_{i},\tregrf_{i})$ is 
\begin{align*}
K_{\bregrf}(T_i,\bprune_{i},\tregrf_{i},\ud x) 
&= \frac{\ind(\bprune_{i}\neq 0)}{|A(T_i,\bprune_{i},\tregrf_{i})|}\sum_{n \in A(T_i,\bprune_{i},\tregrf_{i})}\delta_{n}(\ud x) + \ind(\bprune_{i}=0)\delta_{0}(\ud x),
\end{align*}
where 
\begin{align*}
A(T_i,\bprune_{i},\tregrf_{i}) 
=  \left\{n\in  \{1,\ldots,2N-1\} \setminus \{\bprune_{i}\}:\tregrf_{i}\in(t_{i,n},t_{i,\pa(n)})\right\}.
\end{align*}
%
The variables $\bprune_{i}$, $\bregrf_{i}$, $\tregrf_{i}$ specify an SPR operation that applies to a coalescent tree $T_{i}$ to determine the next one 
which we denote by $\SPR(T_{i},\bprune_{i},\bregrf_{i},\tregrf_{i})$, hence we formally write
\begin{align*}
K_{\canon,\ts}(T_{i},\bprune_{i},\bregrf_{i},\tregrf_{i},\ud x) = \delta_{\SPR(T_{i},\bprune_{i},\bregrf_{i},\tregrf_{i})}(\ud x).\label{eq:canon kernel} \end{align*}

\subsection{Observation process}

Information about the tree process $T=(T_i)_{1\le i\le\maxsite}$, is obtained through the data $D = (D_{i,j}) = (D_{i})$, with columns $D_{i}\in\{0,1\}^{N}$. The mutation model generating $D$ given $T$ is a discrete version of the infinite-sites model \cite{griffiths_et_marjoram97}, and involves a  mutation rate parameter $\theta>0$.
%
Given $T_i$, for a site $1\le i\le S$, the model specifies  that with probability $\exp(-\theta L_i)$ no mutations occur at site $i$, for $L_i$ is as defined in \eqref{eq:tbl_and_bl}, otherwise exactly one mutation arises, and the place where the mutation occurs is uniformly distributed over the branches of $T_i$. 
If a mutation has occurred at a site $i$, we can now infer from $D_i$ the branch on which it arose: it is the unique branch $b_i$ such that $D_{i,n}=1$ if either $n=b_i$ or $n$ is a descendant of $b_i$, otherwise $D_{i,n}=0$. 
The assumption of at most one mutation per site plays a role in our method as it implies that the sequences 
of trees inconsistent with it are not permitted. Under this assumption, the likelihood at site $i$ is
\begin{align}
\!\!\ls_{i}(T_{i}) = P(D_{i}|T_{i})
&= 
\begin{cases}
e^{-\theta L_{i}}, & \sum_{n}D_{i,n} = 0,\\
\left(1{-}e^{-\theta	 L_{i}}\right)\frac{l_{i,b_{i}}}{L_i}, &\sum_{n}D_{i,n} >0.
\end{cases}
\end{align}
Given $T = (T_i)_{1\leq i \leq S}$, mutations occur independently at each site, so the joint likelihood over any set of sites is the product of $\ls_{i}(T_{i})$ over $i$ in the set. 
If $\ls_i(T_{i}) > 0$, then we say that $T_i = (\canon_i,\ts_i)$ is compatible with $D_i$, and since the compatibility of $T_i$ depends only on the topology $\canon_i$ and not on $\ts_i$, we also say that $\canon_i$ is compatible with $D_i$. If $\canon_{i}$ is compatible with $D_i$ for all $i\in\{1,\ldots,S\}$, then we say that $T$ is compatible with $D$.

\section{Tree-bridging MCMC algorithm}
\label{sec:Bridge sampler}

We develop a Metropolis-Hastings 
MCMC algorithm for sampling from $\pi=P(T|D)$, the distribution of the tree process $T$ over a genome interval, given the observation process $D$. 
Convergence of the algorithm is assured from any initial state for $T$, given regularity conditions on the chain (e.g.~Harris recurrent, see e.g.~\cite[page 221]{robert04}). 
However, a choice of initial state that is realistic given $D$ can give an important reduction in the time to convergence. One way to generate a good initial state is to use an existing algorithm from the literature to obtain a point estimate ARG by seeking to minimise the number of recombinations. Arbores uses a variant of the SHRUB algorithm of \cite{song_et_al05} (see also \cite[page 312]{gusfield14}) for this purpose.

\subsection{Segment selection}
The bridges are constructed over genome segments, which can be selected in various ways. All sites should be non-fixed, i.e.~not the conditioning sites, in at least one segment, allowing them to vary. The two end sites of the genome must be treated separately as they will involve only one-sided conditioning. By default, Arbores defines the segments so that, for a chosen $m\ge 1$, the first segment stops at and conditions upon the $(m+1)$th \emph{segregating site} (site where at least one 0 and one 1 are observed). Other segments contain $2m+1$ consecutive segregating sites (including the left/right conditioning sites, which are always segregating) of which the leftmost $m+1$ are shared with the preceding segment. The last segment (with a left-side only conditioning) includes $m+1$ segregating sites from the preceding segment, plus up to $m$ additional segregating sites to reach the right end site of the genome. An example of the segment selection procedure is given in Figure \ref{fig:bridge segments}. More formally, the segments are defined as follows.
Let $\ns$ denote the number of segregating sites and let $s = (s_{1},\ldots,s_{\ns})$, be the indices of the segregating sites in a strictly increasing order. We call $2m+1$ the \emph{bridge length} and assume that $m \ll \ns$. 
For each $j$th segment, where $j=0,\ldots,\jmax$, the initial and terminal site indices $(\alpha_j,\beta_j)$ are defined as follows. For $j=0$, we set $(\alpha_0,\beta_0) = (1, s_{m+1})$, and thereafter 
\begin{equation*}
\alpha_{j} = s_{m(j-1)+1}, \quad \beta_{j} = s_{m(j+1)+1},\quad 1\leq j < \jmax,
\end{equation*}
where $$\jmax \defeq \max\left\{j \in \N:mj +1<\ns \right\},$$ and 
$(\alpha_{\jmax},\beta_{\jmax}) = (s_{m(\jmax-1)+1},\maxsite)$. 
Figure \ref{fig:bridge segments}  shows an example illustration of the deduced segments 
for a case when $m = 4$ and $\ns=20$. 

\begin{figure}
\begin{center}
\includegraphics[width=.9\textwidth]{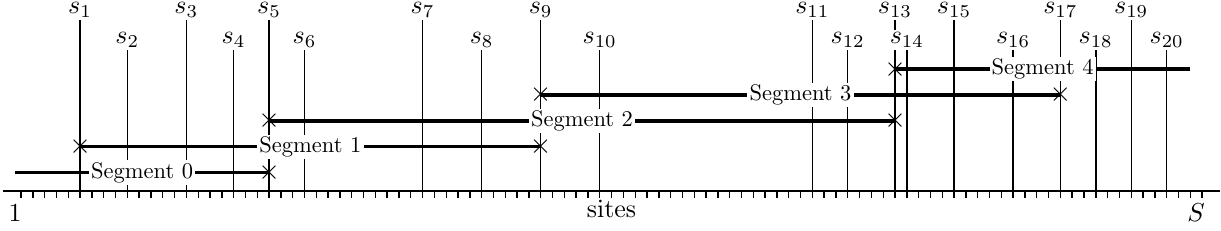}
\end{center}
\caption{Segregating sites $s_{1},\ldots,s_{20}$ are highlighted. Sites within $j$th segment ($j=0,\ldots,4$) are illustrated with conditioning sites marked as $\times$.}
\label{fig:bridge segments}
\end{figure}

\subsection{Generation of bridge proposals}
\newcommand{\SPRc}{\mathcal{S}_{\mathrm{c}}}
\newcommand{\SPRt}{\mathcal{S}_{\mathrm{t}}}

We now focus on a specific bridge, so let $T$ and $\pi$ here denote the restrictions of the tree process and target distribution onto this chosen bridge. We denote by $q(T)$ the proposal distribution for the bridge tree process given its current state. 
A bridge is updated by first sampling $T' \thicksim q(T)$, and then, with probability 
\begin{equation*}
\min\left(1,\dfrac{\pi(T') q(T)}{\pi(T)q(T')}\right),
\end{equation*}
the current state $T$ is replaced with $T'$; otherwise the current bridge $T$ is retained.
We now describe the mechanism for generating bridge proposal from $q(T)$. Each quadruple $(u_i,v_i,r_i,w_i)$ specifying an SPR operation for coalescent tree $T_i$, $L\le i\le R$, must either be equal to $(0,0,0,0)$, representing the identity operation, or satisfy the following conditions: i) $u_i$ is not the root, ii) $v_i\notin\{u_i,\pa(u_i)\}$ (this follows simply by our definition of $v_i$), iii) $w_i>r_i$. 


\subsubsection{Step 1: Tree scanning}
\label{sec:tree scanning}

The first step is to generate the set of all possible sequences of tree topologies over the bridge, compatible with the data $D$ and the assumption of at most one recombination between adjacent sites. We start with $\canon_{L}$, the topology of the left conditioning tree $T_{L}$, and subsequently construct a sequence, or a path, of topologies $(\canon_{i})_{L<i\leq R}$ according to all possible choices of SPR operation nodes, $(u_{i},v_{i})$. Paths that include a topology not compatible with the data are discarded. For the first (leftmost) segment, we have conditioning from the right but not from the left, and so we proceed in reverse direction, from right to left, but otherwise the construction is the same. The tree scanning step is theoretically straightforward, but can be computationally intensive. Each topology at site $i$ can give rise to $\mathcal{O}(N^2)$ topologies at site $i+1$ (see \cite{song06}); the complexity of constructing these topology paths over the segment can be exponential in its length. Heuristics to make the algorithm practically feasible are described later in Section \ref{sec:enabling heuristics}.

\subsubsection{Step 2: Time adjustment}
\label{sec:time adjustment}

By construction, the topology paths generated in Step 1 are consistent with $\canon_L$, but many of them may not be consistent with $\canon_R$ (this is not required for the rightmost segment, for which this time adjustment step can be skipped). Given a topology path $(\canon_{i})_{L\le i \le R}$, one can associate any non-identity transition from $\canon_{i}$ to $\canon_{i+1}$ with the deletion of a node in $\canon_{i}$ and the creation of a new one in $\canon_{i+1}$. We view this pair of events as a \emph{move} of a single node: the composition of all such moves from site $L$ to site $R$ for a node make up the \emph{trajectory} of that node. Thus, when $\canon_{i}\neq \canon_{i+1}$, one node in $\canon_{i}$ moves to a new position, while all other nodes are unchanged. 
Figure \ref{fig:plot_traj} shows an example of the node trajectories from $L$ to $R$. The dashed lines depict node trajectories over consecutive sites. The leaf nodes never move and are therefore omitted from Figure \ref{fig:plot_traj}.

\begin{figure}
\begin{center}
\includegraphics[width=0.85\textwidth]{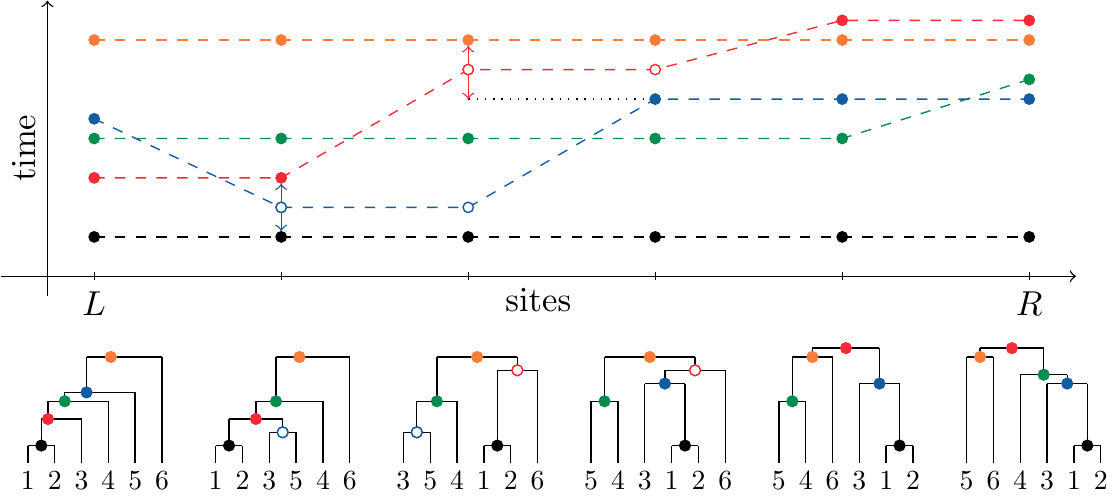}
\end{center}
\caption{Node trajectories for $N=6$ sequences and 5 recombinations. The underlying tree path is shown below the axes. Here $\kappa^{-}(L{+}2,8)=\kappa^{+}(L{+}2,8)=1$ and $\kappa^{+}(R{-}2,9)=1$, but $\kappa^{-}(R{-}2,9)=0$. The nodes affixed by $\ts_L$ and $\ts_R$ are depicted by $\bullet$, {\color{myred}$\bullet$}, {\color{myblue}$\bullet$}, {\color{mygreen}$\bullet$}, and {\color{myyellow}$\bullet$}, while the free nodes depicted by {\color{myred}$\circ$} and {\color{myblue}$\circ$}.}
\label{fig:plot_traj}
\end{figure}

For each site $i$ and non-leaf node $n$, we define the indicators $\kappa^{-}(i,n)$ and $\kappa^{+}(i,n)$ so that $\kappa^{-}(i,n)=1$ (resp.~$0$) if node $n$ at site $i$  moved (resp.~did not move) at the sites from $L$ to $i$. Indicators $\kappa^{+}(i,n)$ are defined in a similar manner, but referring to sites from $i$ to $R$. 
Given these   definitions we proceed as follows. First, we remove paths that do not terminate at $\canon_R$. Then, for each path $(\canon_{i})_{L \leq i \leq R}$ we identify the subset $\cF$ of 
$ \{L,\ldots,R\}\times \{N{+}1,\ldots, 2N{-}1\}$ defined as
\begin{equation}
\label{eq:affixed}
\cF = \big\{(i,n): \kappa^{-}(i,n)=0 \text{ or }\kappa^{+}(i,n)=0 \,\big\}.
\end{equation}
We now take the times $\ts_L$, $\ts_R$ of the left and right conditioning trees, ${T}_L$ and ${T}_R$, also under consideration. All nodes in $\cF$ will be affixed to some time instance in  $\ts_L$ (if $\kappa^{-}(n,i)=0$) or $\ts_R$ (if $\kappa^{+}(n,i)=0$). The nodes in $\cF$ are termed \emph{affixed} while all the other nodes are \emph{free}. Free nodes can only appear in trajectories involving at least two node moves, otherwise they will be affixed to a time either in $\ts_L$ or $\ts_R$.

\begin{figure}
\includegraphics[width=\textwidth]{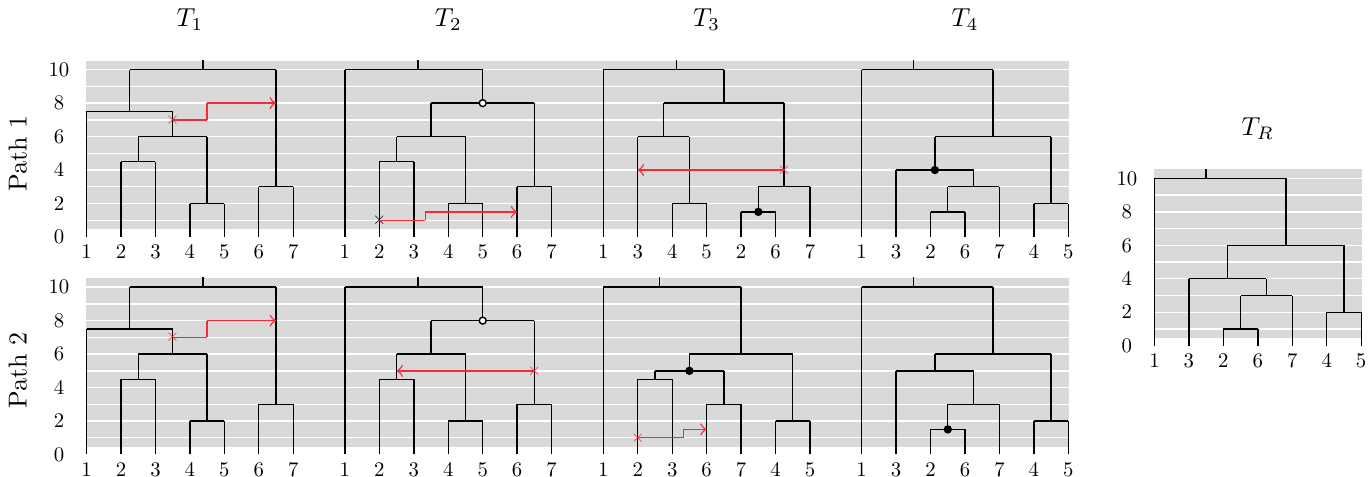}
\caption{Two paths and their generating SPR operations. Path 1 is compatible with the right conditioning tree but Path 2 is not. For Path 2, the conditioning requires the node at time $5$ to be adjusted to time 4, but this cannot be done as it is clear from the second tree of   Path 2, that the node created by the second SPR operation must be in the interval $(4.5,6)$. Hence Path 2 is not compatible with $T_L$, $T_R$. Nodes introduced by the SPR operations that are affixed are depicted by $\bullet$ and free nodes are depicted by $\circ$.
} 
\label{fig:time adj}
\end{figure}



Not all topology paths will be compatible with ${T}_{L}$ and ${T}_{R}$. In practice, we iterate over the entire set of topology paths and verify for each individual path whether the free nodes can be affixed according to ${T}_{L}$ and ${T}_{R}$; if not, the path is discarded. Figure \ref{fig:time adj} shows two paths: one can be affixed while the other cannot. 

\subsubsection{Step 3: Sampling}
\label{sec:sampling}

After Step 2, we have a set of topology paths that are compatible with 
the two-sided conditioning. To fully specify the bridge proposal we sample one of these paths uniformly at random, then generate the times of its free nodes and the pruning times.
In Figure \ref{fig:plot_traj}, there are four free nodes but two free times to sample as the four nodes are comprised of two pairs, the paired nodes being positioned at the same time instance. There can be at most one free time per site and we generate them from left to right. The vertical two-headed arrows in Figure \ref{fig:plot_traj} depict the sampling domains for this example. Notice that for the domain of the free node at site $L{+}2$ one must also consider the affixed nodes on both sides, so that the obtained sequence of coalescent trees is compatible with the orderings implied by the chosen topology path. 

To provide an explicit formula, suppose that $t_{i,n}$ has been identified as a free time, the first subscript $i$ designating the site and the second subscript $n$ designating the node. Once $t_{i,n}$ is affixed, we let $l\in\{0,\dots, R-i-1\}$ denote the number of subsequent 
sites where this node will remain at the same time, and $ 
n_1, \ldots, n_l\in\{N+1,\ldots, 2N-1\}$ 
the labelling of this node at these next sites.
The sampling domain of 
$t_{i,n}$ is then $\mathcal{D}(t_{i,n}) = [t_{i}^{\downarrow},t_{i}^{\uparrow}]$, where 
\begin{align*}
t_{i}^{\downarrow} &\defeq \max\{t_{i,n-1}, t_{i+1,n_1}^{\downarrow},\ldots, t_{i+l, n_l}^{\downarrow}\}, \\
t_{i}^{\uparrow} &\defeq \min\{t_{i,n+1},t_{i+1,n_1}^{\uparrow},\ldots, t_{i+l, n_l}^{\uparrow}\},
\end{align*}
and $t_{i+s,n_s}^{\downarrow}$ (resp.~$t_{i+s,n_s}^{\uparrow}$)
denotes the largest (resp.~smallest) affixed time below (resp.~above) node $n_s$ at site 
$i+s$, for $1\le s \le l$. Note that some of the times involved in the $\min$ can be $+\infty$ if there is no node with affixed time above the node of interest, so the domain $\mathcal{D}(t_{i,n})$ may or may not be bounded from above. 
For bounded domains we choose a uniform distribution and for unbounded domains we choose the shifted exponential distribution with density $\exp(-x+s)$, for all $x > s$ where $s = \inf \{x\in\mathcal{D}(t_{i,n})\}$.
The sampling of the pruning times is easier as, once all nodes are affixed, we simply have the sampling domains
\begin{equation*}
\mathcal{D}(r_i) = \big[\, t_{i,u_i}, \,\min ( t_{i,\pa({u_i})}, w_i  )\,\big], \quad L \le i < R.
\end{equation*}
Again, we choose a uniform distribution. The generation of the bridge proposals is thus completed.



\subsection{Acceptance probability}

We introduce the notation $\tilde{T}_i = \{T_{i},\bprune_{i},\bregrf_{i}, \tprune_{i}, \tregrf_{i}\}$ to refer to the coalescent tree $T_i$ 
endowed with its pruning and regraft 
nodes and times. 
We write $T'=(\tT'_L,\ldots,\tT'_R)$ for the bridge proposal and $T=(\tT_L,\ldots,\tT_R)$ for the current state, 
with $T_L=T'_L$, $\tT_R=\tT'_R$. As the tree process is Markovian and the observations independent (given the trees), the target distribution ratio in the Metropolis-Hastings acceptance probability can be written as
\begin{equation}
\dfrac{\pi(T') }{\pi(T)}
= \frac{\prod_{i=L}^{R-1}K(\tT'_{i-1},\tT'_{i})\, \prod_{i=L+1}^{R-1}\ls_{i}(T'_{i})}
{\prod_{i=L}^{R-1}K(\tT_{i-1},\tT_{i})\,\prod_{i=L+1}^{R-1}\ls_{i}(T_{i})},\label{eq:target ratio}
\end{equation}
under the convention  that $K(\tT_{0}',\tT_{1}') = \mu 
(\tT_{1}')$. 
 %
Let $Q_{\cP}$ denote the proposal probability mass function of the topology paths after Step 2,
let $(\canon_L',\ldots, \canon_R')$ denote the topology paths corresponding to $T'$, and let $M'$ denote the number of 
free times $(f'_{1},\ldots, f'_{M'})$ for $T'$
with joint density $Q_{f}(f'_{1},\ldots, f'_{M'})$. Also, we let  $Q_r(r'_L,\ldots, r'_{R-1}\,|\,f'_{1},\ldots, f'_{M'})$ denote the conditional joint density for the pruning times. The corresponding quantities for the current position are defined analogously in the obvious way by omitting the prime. We complete the specification of the acceptance probability in our bridging 
algorithm via the calculation
\begin{align}
\dfrac{ q(T)}{q(T')} 
=
\dfrac{Q_{\cP}(\canon_L,\ldots,\canon_R)}{Q_{\cP}(\canon'_L,\ldots,\canon'_R)}\times
\dfrac{Q_f(f_{1},\ldots, f_{M})}{Q_f(f_{1}',\ldots, f'_{M'})}
\times \dfrac{Q_r(r_L,\ldots, r_{R-1}\, |\,f_{1},\ldots, f_{M} )}{\,Q_r(r_L',\ldots, r'_{R-1}\,|\,f_{1}',\ldots, f'_{M'})} .
\label{eq:acceptance probability}
\end{align}

The above algorithm corresponds to an Independence 
Metropolis-Hastings algorithm. It is easy to check  
that, under reasonable choices for the proposal, $\sup_T \pi(T)/q(T)<\infty$, with the supremum taken over the support of $\pi(\cdot)$, thereby guaranteeing uniform ergodicity for the bridge sampler. To see that, notice that the number of the permitted topology paths is finite, so one only needs to assign non-zero probability to each one of them (e.g.~via the discrete uniform mentioned earlier in the text); then, the number of free regraft and pruning times is also finite, and one needs to select a lower bounded density for the times of finite 
support (e.g.~the continuous uniform referred to  earlier) and a proper density for the 
unbounded times (e.g.~the `prior' exponential density chosen above will dominate the posterior density). 
In terms of the complete method, the use of overlapping blocks 
implies that uniform ergodicity will also hold as long as all
topologies over the complete genome supported by the posterior 
can be visited by the proposal during the execution of the algorithm. We conjecture that this is true due to the flexibility of the 
method, but a rigorous proof requires elaborate work 
on a theme exceeding the scope of the paper. 
Such ergodicity results are of qualitative nature, and the efficiency 
of the method is determined by more practical considerations, e.g.~the computing cost for the realisation of the proposal or the size of the acceptance probability.

\subsection{Coalescent time sampling}

Arbores implements an additional MCMC   step that proposes the movement only of the coalescence times. This time sampling is scheduled for execution after each complete execution of all bridge segment steps. The sampling is done via a standard Metropolis-Hastings step for each coalescence time separately in a manner that preserves the coalescence time ordering throughout the entire sequence. The sampling distribution used in this step is a truncated version the conventional exponential distribution for coalescence times (see e.g.~the discussion following equation (4) in \cite{hudson90}). Truncation is needed to ensure the preservation of the coalescent time ordering.

\section{Enabling heuristics}
\label{sec:enabling heuristics}

The above `idealised' algorithm outlined can be too computationally expensive to be implemented in practice. This is mainly due to the worst case exponential size of the set of topology paths in the number of sites within the bridge segments. 
Therefore, we need to introduce heuristics, some of them approximative, to deliver a practical algorithm. The main principle underlying the choice of heuristics is the one of \emph{minimum-recombination}.
That is, we allow the data to enforce recombinations when required, and we avoid placing recombinations at positions not supported as such by the data. At the same time we require the algorithm to traverse the space of different recombinations, at different sites, so we perform another heuristic to also allow for this. One can think of our algorithm as one that follows principles from deterministic minimum-recombination algorithms used in a separate stream of the subject literature
to reduce computing costs, while still being a proper (if approximate) posterior sampling MCMC algorithm. Due to this minimum recombination approach, we expect the approximations to improve as the ratio of mutation rate over recombination rate per site increases. We also wish to acknowledge that, due to some of the heuristics, our approximation is not guaranteed to be reversible. However, problems due to the irreversibility are expected to be rare and hence this is considered an acceptable approximation. The heuristics are described below.

\subsection{Parsimonious SPR operation}
\label{sec:necessary recombinations}

The first heuristic aims to reduce the cardinality of the set 
of proposed tree paths, by switching from an exhaustive tree scan
to one that adopts a parsimonious approach as regards to the number of SPR operations, taking under consideration the information available in the data.
%
%
%
 
For non-segregating sites any tree is compatible, so in the generation of the topology paths we omit non-segregating sites and
perform SPR operations only between consecutive segregating sites -- if necessary, i.e.~if none of the currently generated topologies is compatible with the next segregating site.
This leads to a substantial reduction in the computational cost.
Some further considerations are needed here, as it may well be the case that more than one SPR operation is needed between two  segregating sites to generate trees compatible to the data. 
Thus, we iterate -- if necessary -- over the number of SPR operations.
In practice, Arbores attempts to construct the topology paths with none, one, or (at most) two SPR operations between each pair of consecutive segregating sites. Topologies available at the current 
step that require more SPR operations than others are removed (their path is deleted from the set of currently constructed paths).
 In the numerical applications we have looked at, requiring more that two SPR operations between consecutive segregating sites was a rare event; in the cases that 
this does occur, Arbores skips the update on that segment. Even if a segment update is skipped at an iteration of the MCMC sampler, it is likely that due to the updates of the other segments, the skipped segment can be updated the next time it is processed.

\subsection{Extra recombinations }
\label{sec:extra recombinations}

The parsimonious rules for allowing SPR operations described above come with a drawback: for a given bridge, by adopting, effectively, a minimal number of recombinations principle, one can  deduce that all generated topology paths will have the same number of recombination sites along the bridge and, in fact,  the same number of recombinations  between all pairs of segregating sites contained in the bridge.
To overcome this rigidity, we allow in a controlled way the topology paths to consist of more than the minimum number of SPR operations. This is done by first performing the tree scan as described in part (a), resulting in a set of tree topology paths. After this we repeat the tree scan in a sightly modified way. For the first pair of consecutive segregating sites that does not require an SPR operation, we nevertheless introduce one. For the remaining pairs of segregating sites the tree scan proceeds normally by introducing SPR operations only if needed. The resulting topology paths are added to the set of previously generated paths. We repeat this modified tree scan step for each remaining pair of consecutive segregating sites where no SPR is required. In this way we allow only one extra SPR operation for one pair of consecutive segregating sites at a time, with the aim of keeping the combinatorics manageable. At the same time, these additional SPR operations are enough to allow mobility in the number of recombinations and their positions.


\subsection{Subtree search}
\label{sec:subtree search heuristics}

The computationally most expensive part of the algorithm, even after implementing the heuristics described above, is the iteration over the possible SPR operations in the construction of the topology paths. In a naive implementation, one would simply apply all possible sequences of SPR operations and then check for each individual outcome whether the resulting tree is compatible with the data or not. Some computation can be saved by observing that certain operations are known in advance to produce an incompatible tree. 
We identify the operations that may produce a compatible tree as follows. Consider a tree at a segregating site $i$. We assign colours (black or white) to its leaf nodes, by specifying that the leaf $n$ is black if the $n$th observation at the \emph{next} segregating site, say $j$, equals 1 and white otherwise. Note, that we are considering the data at the next segregating site $j>i$ because the aim is to characterise the operations that produce a tree compatible with the data at site $j$. The colouring is extended to all nodes by recursively defining the colour of a node to be equal to the common colour of its children, if such colour exists, and grey otherwise, as demonstrated in Figure \ref{fig:heuristic operations}. Note also that the role of black and white nodes is not interchangeable, because for a tree to be compatible with the data, it must contain exactly one subtree whose all nodes are black and which is not a subtree of another subtree with black nodes only. Such condition is not imposed on the white nodes.
We then have to consider only two types of SPR operation nodes $(u_i,v_i)$:
\begin{itemize}
\item[(i)] $u_{i}$ is black and $\pa(u_{i})$ is not black and $v_{i}$ is black,
\item[(ii)] $u_{i}$ is white and $\pa(u_i)$ is not white and $v_{i}$ is white, gray or black and $\pa(v_{i})$ is not black.
\end{itemize}

\noindent It is not guaranteed that these operations yield compatible trees, but in the case of single SPR operation between segregating sites, it can be proven that the set of operations yielding a compatible tree is a subset of such operations. The proof is included in the Appendix, \mbox{Theorem \ref{thm:heuristic}} in Section \ref{sec:exhaustiveness}. This implies that this heuristic does not introduce any additional approximations. Indeed, one can see in Figure \ref{fig:heuristic operations}, showing all SPR operations yielding a compatible tree, that each of these operations are either of type (i) or (ii) defined above. For more than one SPR operation between segregating sites, this is not true and the reduction of the SPR operation search set to operations of type (i) or (ii) will result in an approximation.

\begin{figure}
\begin{center}
\includegraphics[width=.66\textwidth]{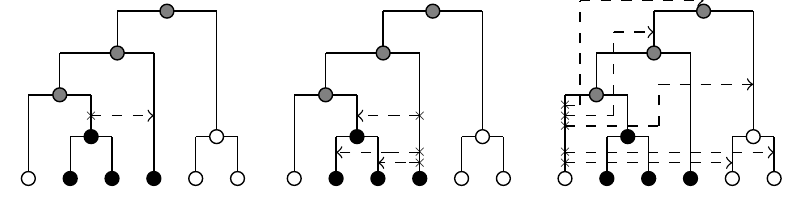}
\end{center}
\caption{All SPR operations that yield a compatible tree. Each operation is of kind (i) or (ii).}
\label{fig:heuristic operations}
\end{figure}

\subsection{Output}
We note that -- through the heuristics -- the algorithm aims to provide a principled approximation to 
the idealised MCMC algorithm defined in the previous sections.
Ultimately, the proposal will provide samples from the support of the idealised posterior:
e.g.\@ when the proposal, using the heuristics, has determined a recombination between two successive segregating 
sites, the non-segregating sites are also taken under consideration and the recombination site
is selected in some manner (e.g.~at random) amongst the intermediate non-segregating sites and the right-side segregating one, yielding a topology path over all sites.
%

\section{Numerical experiments}
\label{sec:numeric}

We carried out two numerical experiments with the proposed algorithm. In the first experiment, we ran the algorithm on the well-known Kreitman data~\cite{kreitman} and in the second experiment we carried out a comparison with the 
recently proposed ARGweaver algorithm \cite{rasmussen_et_al14}, sometimes perceived as the state-of-art method for this problem. 

\subsection{Kreitman data}
\label{kreitman}

The  Kreitman data were first preprocessed by removing duplicate rows and columns with minor allele count less than two. The resulting data consisted of DNA polymorphisms of 9 sequences across 2,287 sites of which 30 were segregating. We set the mutation parameter to $\theta =  0.013$ corresponding to approximately 30 mutations on average for the data of the given size. The recombination rate was set to $\rho = 0.0035$ corresponding to approximately 8 recombinations in average. The minimum number of recombinations for these data is known to be 7 (see, e.g.~\cite[page 144]{gusfield14}).

The chain was run for $2\times 10^5$ iterations. Iteration here means either the processing of a single segment or sampling a single coalescent time. This number of iterations amounts to all segments having been sampled approximately 6,500 times, the exact number slightly varying according to the number of recombinations at a given state of the chain (see the discussion at the end of Section \ref{sec:enabling heuristics}\ref{sec:necessary recombinations}). To confirm the consistency of the results, the algorithm was run four times independently. 

\begin{figure}[!t]
\includegraphics[width=.95\textwidth]{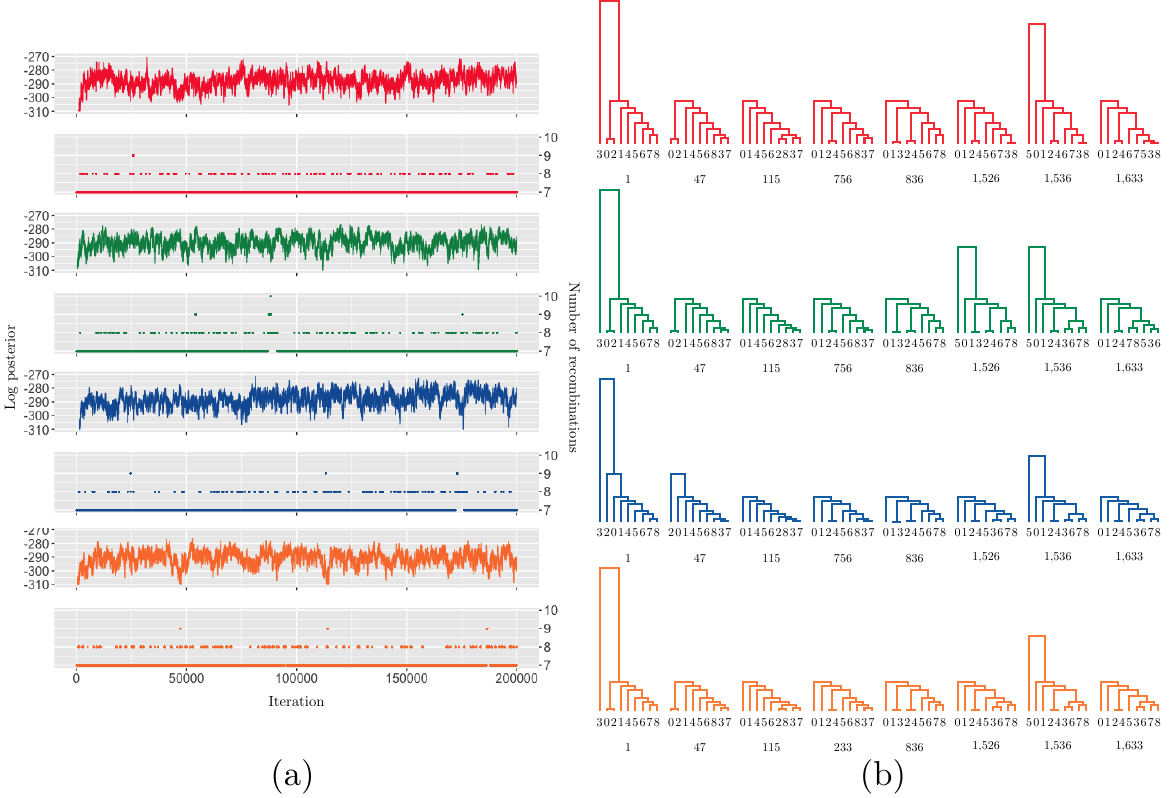}
\caption{(a) The unnormalised log posterior and the number of recombinations trace plots for four independent runs. (b) The maximum a posteriori   tree sequences of four independent runs. Coalescence times are calculated as the averages of coalescence times over all tree sequences with matching structure. The site of the first occurrence of a given tree is shown below each tree.}
\label{fig:log posterior}
\end{figure}
Figure \ref{fig:log posterior}a shows the trace plots of unnormalised log posterior densities accompanied with the trace plots for the number of recombinations. Each of the chains appears to spend most of the time sampling ARGs with the minimum number of recombinations, but occasional visits to up to 10 recombinations can be seen. 

Figure \ref{fig:log posterior}b shows the maximum a posteriori (MAP) tree sequences for the four runs. The coalescent times of these MAP trees are the sample averages over the sequences with matching tree structure. Each of these sequences corresponds to an ARG. One can see that although the MAP ARGs are similar, they are not identical. This leads us to believe that in the posterior, there are ARGs with slightly different structure but approximately same posterior probability, hence the MAP ARGs in different runs do not have to be precisely the same (sampling variation might also have an effect). Each of the MAP ARGs is displaying 7 recombinations which is consistent with the number of recombinations trace plots in Figure \ref{fig:log posterior}.

The algorithm was implemented in C and is available at \url{https://github.com/heinekmp/Arbores}. The running time of the algorithm is random. For the numerical experiments reported here, the running time was approximately 10 hours on an off-the-shelf MacBook Pro 
(2.9 GHz Intel Core i7).

\subsection{Comparison with ARGweaver}

In our second experiment, we run both Arbores and ARGweaver on the same simulated data. The data were first generated with MaCS software of~\cite{chen_et_al09} after which it was preprocessed by removing sites with minor allele count less than two; the data ware also shifted so that first segregating site was given index one. After preprocessing, the data consisted of eight sequences across 10,000 sites of which 24 were segregating. The datasets and the exact calls of the algorithms are available at \url{https://github.com/heinekmp/Arbores/tree/master/test_runs}.

We compare the outputs of the algorithms against each other rather than against the true ARG that was used for generating the data. This is done because with  Bayesian MCMC methods the actual target is the posterior distribution which may or may not be an accurate representation of the generating ARG. 

\begin{figure}
\begin{center}
\includegraphics[width=\textwidth]{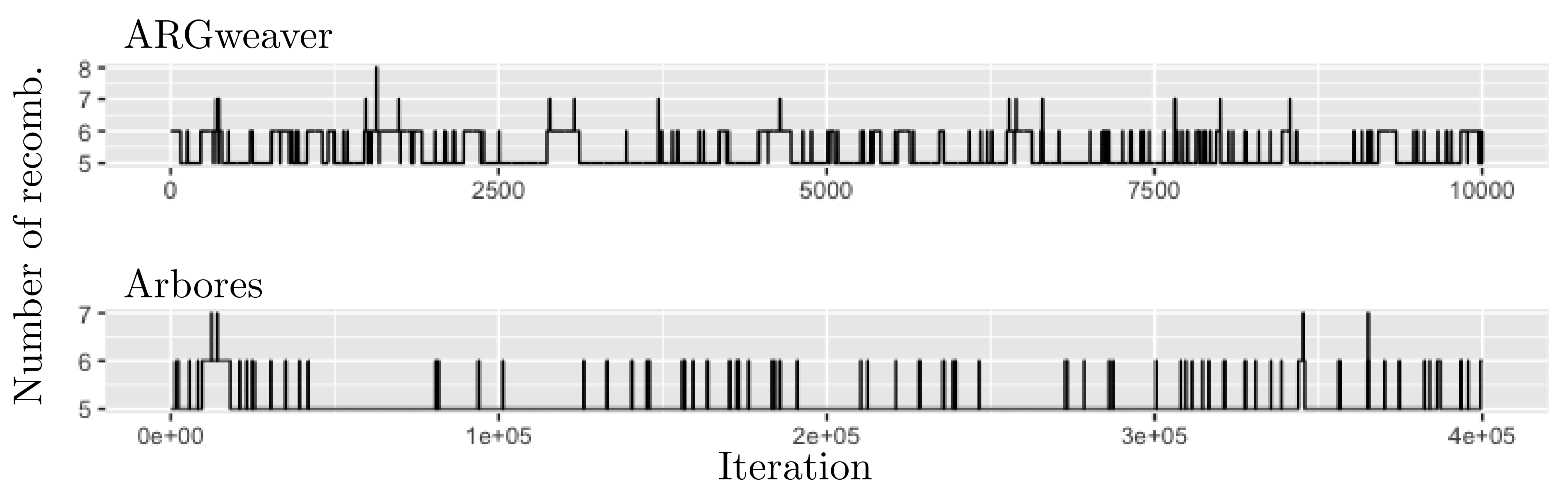}
\end{center}
\caption{Trace plots of the number of recombinations for ARGweaver and Arbores. For Arbores the horizontal axis shows all iterations which amounts to all segments being sampled roughly 15,000 times.}
\label{fig:nrecomb}
\end{figure}

Figure \ref{fig:nrecomb} shows the trace plots for the number of recombinations for both algorithms. The initialisation of Arbores aims at starting with the minimum recombination ARG but this is not guaranteed. The results nevertheless seem to suggest that the minimal number or recombinations for this dataset is 5. In some simulations ARGweaver sampled ARGs with fewer than 5 recombinations but in all such cases the sampled ARG was not compatible with the data. It is worth pointing out that Arbores never returns an ARG that is not compatible with the data. From Figure \ref{fig:nrecomb} we see that ARGweaver mixes somewhat better between different numbers of recombinations although the trace plots are similar. This may be due to genuinely different mixing properties, or that the algorithms target different posteriors.

\begin{figure}
\includegraphics[width=\textwidth]{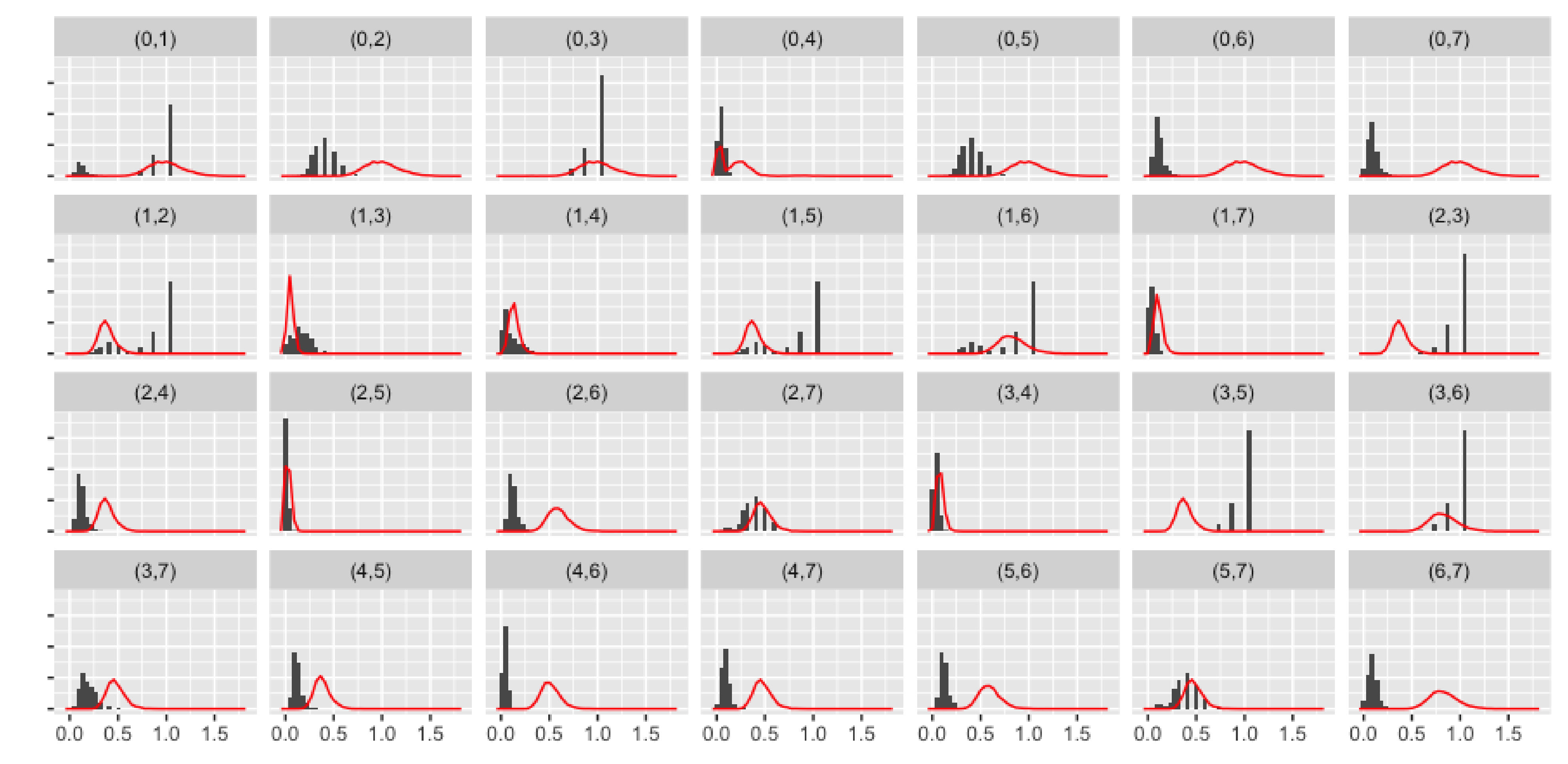}
\caption{Histograms of TMRCAs for Arbores (red) and ARGweaver (black). See the text for the definition of TMRCA. The title on each histogram indicates the pair of observed sequences indexed by $\{0,\ldots,7\}$.}
\label{fig:hists}
\end{figure}

In order to compare the resulting approximate posterior structure of the ARG, Figure \ref{fig:hists} shows the histograms of the times to most recent common ancestor (TMRCA) for both algorithms. Note that to better capture   the structure of the ARG we use a slightly non-standard definition of TMRCA. TMRCA was calculated for each pair of observed sequences (indexed by $\{0,\ldots,7\}$) by calculating the minimum time to the most recent common ancestor at each site and then by taking the minimum over all sites. TMRCAs were also scaled so that the greatest mean over all histograms for both algorithms is equal to one.

Our first observation from the results reported in Figure \ref{fig:hists} is the impact of the time-discretisation adopted within ARGweaver. Time-discretisation is not required for the algorithm suggested in this paper, whereas for ARGweaver it is a necessity as the algorithm develops on a finite state-space for ARGs. For both algorithms the histograms are calculated similarly with 40 equally spaced bins, but only for Arbores the histograms appear to represent the densities of  continuous distributions as desired. Particularly for the pairs with large mean, e.g.~(3,6), ARGweaver returns truncated one tailed distributions which is not an accurate representation of the reality. Adjustment of the parameterisation, e.g.~the maximum time or higher resolution of the time discretisation, are an obvious remedy to this issue, but due to the logarithmic scale of the time discretisation in ARGweaver, a large number of discrete times would be required to allow high resolution at large times, which will slow down the algorithm: doubling the number of time-discretisation points from 40 to 80, the computation time of 10,000 samples would increase from 1.5 hours to around 7.5 hours. Arbores took approximately 10 hours to generate 15,000 samples. The computing equipment used was as described in Section \ref{sec:numeric}\ref{kreitman}.

We also see that for some pairs of sequences, both algorithms identify a very recent common ancestor, see the pairs (1,3), (1,4), (1,7), (2,5) and (3,4). For some pairs both algorithms also agree on more distant common ancestor, see the pairs (0,1), (0,3), (1,6) and (3,6). Other histograms suggest that one of the algorithms was not able to explore all the modes, see e.g.~pairs (0,4), (1,5) and (1,6). Further discrepancies between the histograms can be explained by the fact that ARGweaver had the additional degree of freedom to decide which of the two alleles appearing at a given segregating site was the mutated one while for Arbores the data determines the mutated alleles; data entries equal to one are mutated. While such flexibility may sometimes be desired, it also increases additional variation to the results in cases where the mutated alleles are known. 


\section{Concluding remarks}
\label{sec:conclude}

We have proposed an MCMC algorithm for simulating ARGs from the Bayesian posterior distribution given observed DNA polymorphism data. The algorithm is based on a novel bridging procedure which enables us to reduce the high dimensional problem into a set of substantially smaller scale problems. The main benefit of the algorithm is its suitability for parallel computing systems. This is due to the fact that to some extent the bridge segments can be processed in parallel independently of each other. 

Further research is still needed to improve the scalability of the algorithm in the number of observed sequences. In particular, we believe that substantial improvements can be made by replacing the tree scanning step with more sophisticated methods that reduces the number of discarded paths in the time adjustment step, and thus avoid redundant computations. A potential approach is a node scoring approach whereby nodes are assigned a score between 0 and 1 according to the data at its leaf descendants. Score 1 (resp. 0) would correspond to all leaf descendants assuming value 1 (resp. 0). These scores might be indicative of the compatibility of the tree with the data and could be used for steering the tree paths in a manner which reduces the number of discarded paths. Other aspects of the algorithm warrant further investigation, e.g.~non-uniform 
distributions for the choice of the proposed topology path could be tried. In general, 
Arbores provides a conceptually simple approach for sampling ARGs and as such, 
further potential improvements are seemingly easy to incorporate into its algorithmic framework.

Our comparison with ARGweaver shows that although its flexible parameterisation allows it to be used for more realistic problems than Arbores, the time-discretisation is a limitation which may be manifested already with modestly sized datasets. While a more thorough experimentation with ARGweaver might have lead to a parameterisation that efficiently mitigates the impact of time-discretisation, it can be argued that methods based on modelling continuous phenomena by discretisation, can reach the same accuracy as methods based on continuous models, such as Arbores, only asymptotically.

We have used SMC as an the Markovian approximation to ARG; our methodology can be applied without modification with the more accurate 
 \cite{wilt:15} SMC' approximation \cite{marj:06}.

\begin{description}
\item[Ethics:]The project did not involve research on humans or animals.
\item[Data:]The Kreitman data set is reported in \cite{kreitman} and can also be found in \cite[page 144]{gusfield14}. Other datasets and the exact calls of the algorithms are available at \url{https://github.com/heinekmp/Arbores/tree/master/test_runs}
\item[Contributions:]M.D.I, A.B., A.J., D.B.~planned and designed the project, M.D.I, A.B., A.J., D.B., and K.H.~carried out the research and wrote the manuscript. K.H.~implemented the algorithm and carried out the numerical experiments.
\item[Competing interests:]We have no competing interests.
\item[Funding:] This work was supported by the EPSRC grant ``Advanced Stochastic Computation for Inference from Tree, Graph, and Network Models'' (ref: EP/K01501X/1).
\item[Acknowledgements:]A.B. also acknowledges support from a Leverhulme Trust Prize.
\end{description}

\bibliographystyle{plain}
\bibliography{mybib}

\begin{appendices}

\section{SPR operation search heuristic}
\label{sec:exhaustiveness}

In this section we prove the claim made earlier that the subtree search heuristic introduced in Section \ref{sec:subtree search heuristics} does not introduce any error under the assumption of at most one recombination between consecutive segregating sites.

In addition to the black, grey, white colouring of the nodes described earlier, we introduce another classification of nodes depending on whether a node is the root of a maximal subtree of their respective colour or not. A subtree is said to be black (resp.~white) if all its nodes are black (resp.~white). A black (resp.~white) subtree is said to be maximal if any strictly larger subtree containing it is not black (resp.~white). A grey node cannot be a root of a subtree consisting only of grey nodes, so all grey nodes are classified as branch nodes by convention. 

With this classification together with the node colouring we can construct equivalence classes of SPR operations which we denote by $(x,y,z,w)$, where $x,z\in\{\mathrm{B},\mathrm{W},\mathrm{G}\}$ denoting the colours (black, white, grey) of the pruning and regrafting nodes respectively, and $y,w \in \{\mathrm{r},\mathrm{b}\}$ denoting the classifications (root of a subtree, branch) of the pruning and regrafting node, respectively. Also we use $\ast$ as a wildcard to denote any possible value of a given entry. We have the following theorem:
\begin{theorem}\label{thm:heuristic}
Let the colours black, white and grey be assigned to the nodes of $T_i$ as described above, and assume that $T_i$ contains more than one maximal black subtree. The SPR operation that results in a tree containing at most one maximal black subtree (i.e.~makes $T_i$ compatible with the data), if such operation exists, belongs to $\sprop{B}{r}{B}{\ast}$, $\sprop{W}{r}{W}{\ast}$, $\sprop{W}{r}{G}{\ast}$ or  $\sprop{W}{r}{B}{r}$.
\end{theorem}
\begin{proof}
All possible SPR operations can be expressed as a set of equivalence classes:
\begin{align*}
\begin{array}{c}
\sprop{B}{r}{B}{\ast},~\sprop{B}{r}{W}{\ast},~\sprop{B}{r}{G}{\ast},~\sprop{B}{b}{\ast}{\ast},~\sprop{W}{b}{\ast}{\ast},\\
\sprop{W}{r}{W}{\ast},~\sprop{W}{r}{G}{\ast},~\sprop{W}{r}{B}{r},~\sprop{W}{r}{B}{b},~\sprop{G}{\ast}{\ast}{\ast}.
\end{array}
\end{align*}
Each of the Lemmata \ref{lem:first} -- \ref{lem:last} below, shows that the SPR operations within a specific equivalence class cannot yield a compatible tree. This leaves us with the operations mentioned in the statement of the theorem.
\end{proof}

\begin{lemma}\label{lem:first}
Operations in $\sprop{B}{r}{W}{\ast}$ or $\sprop{B}{r}{G}{\ast}$ do not reduce the number of maximal black subtrees.
\end{lemma}
\begin{proof}
The classification of a node as a subtree root can only change if either $1^\circ$ the colours of its descendants change or $2^\circ$ the colour of its sibling changes. Suppose we have two black subtree root nodes, $b_1$ and $b_2$. Let us first consider how the pruning of $b_{1}$ affects the subtree root status of $b_2$. 

The pruning of any node $u$ can only affect the colours of the ancestors of $u$. Due to being a black subtree root, $b_1$ cannot be a descendant of $b_2$ and therefore pruning $b_1$ cannot affect the colours of the descendants of $b_2$. Therefore the subtree root classification of $b_2$ can only change if the pruning of $b_1$ turns the sibling of $b_2$ black. By considering all possible scenarios we see that after pruning $b_1$ the ancestors of $b_1$ either remain grey or turn white and hence the sibling of $b_2$ cannot turn black due to the pruning of $b_1$. In conclusion, pruning $b_1$ will not change the status of $b_2$ as a black subtree root.

Let us then consider the effects of regrafting the subtree rooted at $b_1$. Regrafting can only affect the colours of the ancestors of the regrafting node. Since $b_1$ is regrafted either to a white or a grey node, it cannot be regrafted to a descendant of $b_2$ and therefore the status of $b_2$ can only change if the sibling of $b_2$ is an ancestor of the regrafting node and if it is black after regrafting. Regrafting a black node to a white or a grey node results in all the ancestors of the regrafting node being grey, so the status of $b_2$ as a black root is unchanged. Moreover, $b_1$ remains as a black subtree root, so the resulting tree has exactly the same number of maximal black subtrees as the original tree. \qedhere
\end{proof}

\begin{lemma}
Operations in $\sprop{B}{b}{W}{\ast}$, $\sprop{B}{b}{B}{\ast}$ or $\sprop{B}{b}{G}{\ast}$ when applied to a tree with more than one maximal black subtree cannot produce a compatible tree.
\end{lemma}
\begin{proof}
We only need to consider regrafting, because after the pruning, the number of maximal black subtrees remains unchanged. From the proof of Lemma \ref{lem:first} we know that regrafting a black node to a white or a grey node cannot change the status of the existing black subtree roots but it introduces a new maximal black subtree. So the resulting tree cannot be compatible. 

Regrafting a black node to a black branch node means inserting a black subtree into a black subtree whose root node remains unchanged. Regrafting a black node to a black subtree root, say $b$, implies that the new node introduced by the regraft operation becomes a new black subtree root and the classification of $b$ changes from subtree root to branch. In any case, the number of black subtree root nodes remains unchanged.
\end{proof}

\begin{lemma}\label{lem:wrbb}
Operations in $\sprop{W}{r}{B}{b}$ cannot produce a compatible tree.
\end{lemma}
\begin{proof}
After the regrafting, the black branch node, say $b$, to which the white node was regrafted becomes a black subtree root. Also the sibling of $b$ becomes a black subtree root. Hence the resulting tree has at least two maximal black subtrees and cannot be compatible.
\end{proof}

\begin{lemma}
Operations in $\sprop{W}{b}{B}{\ast}$, $\sprop{W}{b}{W}{\ast}$ or $\sprop{W}{b}{G}{\ast}$ when applied to a tree with more than one maximal black subtrees cannot produce a compatible tree.
\end{lemma}
\begin{proof}
Pruning a white branch node does not affect the number of maximal black subtrees so we only need to consider regrafting. As in Lemma \ref{lem:wrbb}, regrafting to a black branch node will cause the regrafting node and its sibling to become black subtree roots and thus the tree will not be compatible.

When regrafting to a black subtree root, the black subtree root remains unchanged and all its ancestors become grey. Hence the classification of any black subtree root remains unchanged as a change would require its sibling to turn black which cannot be the case.

Regrafting a white node to a white branch node will have no consequences outside the white subtree containing the regrafting node. Regrafting to a white subtree root causes the new node introduced by the regraft operation to become a new white subtree root instead of the regrafting node, but the rest of the tree will remain unchanged. 

When regrafting to a grey node, only the ancestors of the regrafting node will be affected, but because the ancestors of a grey node are grey, they will remain unchanged. \qedhere

\end{proof}

\begin{lemma}\label{lem:last}
Operations in $\sprop{G}{\ast}{\ast}{\ast}$ when applied to a tree with more than one maximal black subtree cannot produce a compatible tree.
\end{lemma}
\begin{proof}

Regrafting a grey node to a black node (subtree root or a branch) results in a tree containing at least two maximal black subtrees: one rooted at the regrafting node and another one must be contained by definition in the subtree rooted at the pruned grey node.

Regrafting to a white or a grey node will cause all the ancestors of the regrafting node to become grey. This means that the classification of each black subtree root in the tree must remain unchanged, but since the subtree being regrafted contains, by definition, at least one maximal black subtree, the resulting tree must have at least one more maximal black subtree than the tree after pruning. If the tree after pruning contained at least one maximal black subtree, then the resulting tree would be incompatible and if the tree after pruning did not contain any maximal black subtrees, then the pruned subtree must contain at least two maximal black subtrees, since the original tree was assumed to have at least two maximal black subtrees. In any case, the resulting tree will be incompatible.
\end{proof}

\end{appendices}

\end{document}